\documentclass[journal]{IEEEtran}

\usepackage{amsmath,amssymb}
\usepackage{graphicx}
\usepackage{cite}
\usepackage{stfloats}
\usepackage{bm}
\usepackage[caption=false,font=footnotesize,farskip=0pt,captionskip=2pt]{subfig}
\newtheorem{lemma}{Lemma}
\newtheorem{remark}{Remark}

\newcommand{\Nt}{N_{\mathrm{t}}}
\newcommand{\Pt}{P_{\mathrm{t}}}
\newcommand{\bR}{\mathbf{R}}
\newcommand{\ba}{\mathbf{a}}
\newcommand{\aMC}{\ba_{\mathrm{MC}}}
\newcommand{\arob}{\ba_{\mathrm{rob}}}
\newcommand{\bh}{\mathbf{h}}
\newcommand{\bw}{\mathbf{w}}
\newcommand{\bx}{\mathbf{x}}
\newcommand{\bI}{\mathbf{I}}
\newcommand{\bZ}{\mathbf{Z}}
\newcommand{\bZb}{\bar{\mathbf{Z}}}
\newcommand{\bZrob}{\mathbf{Z}_{\mathrm{rob}}}
\newcommand{\bDel}{\bm{\Delta}}
\newcommand{\bPhi}{\bm{\Phi}}
\newcommand{\bPsi}{\bm{\Psi}}
\newcommand{\CN}{\mathcal{CN}}
\newcommand{\bXi}{\bm{\Xi}}
\newcommand{\bQ}{\mathbf{Q}}
\newcommand{\hMC}{\bh_{\mathrm{MC},k}}
\newcommand{\hrob}{\bh_{\mathrm{rob},k}}
\newcommand{\Rea}{\mathrm{Re}}
\newcommand{\epss}{\varepsilon_{\mathrm{s}}}
\newcommand{\epsc}{\varepsilon_{\mathrm{c},k}}

\makeatletter
\g@addto@macro\normalsize{%
  \setlength{\abovedisplayskip}{3pt plus 1pt minus 1pt}%
  \setlength{\belowdisplayskip}{3pt plus 1pt minus 1pt}%
\setlength{\abovedisplayshortskip}{3pt plus 1pt}
\setlength{\belowdisplayshortskip}{3pt plus 1pt}
}
\makeatother

\usepackage{xcolor}

\begin{document}
\title{Robust Beamforming Design for Integrated Sensing and Communications with Mutual Coupling Effect}
\author{Jieon~Maeng,~\IEEEmembership{Student~Member,~IEEE,}
        and~Kawon~Han,~\IEEEmembership{Member,~IEEE}%
\thanks{
This work was supported by funding from Korea government (KASA, Korea AeroSpace Administration) under Grant RS-2026-25578911 and by the Institute of Information and Communications Technology Planning and Evaluation (IITP) grant funded by the Korea Government (MSIT) under Grant RS-2026-25525599.

Jieon Maeng and Kawon Han are with the Department of Electrical Engineering, Ulsan National Institute of Science and Technology (UNIST), Ulsan, South Korea (e-mail: jemaeng@unist.ac.kr; kawon.han@unist.ac.kr).}%
}
\maketitle
\raggedbottom

\begin{abstract}
Integrated sensing and communications (ISAC) is a key technology for next-generation wireless networks, enabling communication and radar sensing over shared spectral and hardware resources. In practical multi-user multiple-input multiple-output (MU-MIMO) ISAC transmitters, however, mutual coupling (MC) between antenna elements distorts the array steering vector and each communication user (CU) channel, so that the sensing beampattern deviates from the desired one and the communication link to each user degrades. To address this limitation, we propose a robust MC-compensated beamforming design that guarantees both the sensing and communication performance of MU-MIMO ISAC transmitters against the residual MC error. 
We introduce an uncertainty model whose center and bound are derived from the physical MC model, where the same underlying residual MC perturbation induces both steering-vector and communication-channel uncertainties. Using tractable norm-bounded approximations of these uncertainties,
the transmit covariance is optimized against the worst-case of each uncertainty, minimizing the worst-case beampattern matching mean-squared error (MSE) for sensing while guaranteeing the worst-case signal-to-interference-plus-noise ratio (SINR) for each CU. Each resulting worst-case constraint is converted into a linear matrix inequality, and the problem becomes a convex semidefinite program (SDP).
Numerical results show that the proposed robust MC-compensated design attains a lower beampattern matching MSE and a higher achieved SINR than the existing robust ISAC beamforming design.
\end{abstract}

\vspace{-5pt}
\section{Introduction}
\IEEEPARstart{I}{ntegrated} sensing and communications (ISAC) is a key technology for next-generation wireless systems, allowing a single transceiver to simultaneously perform communication and radar sensing over shared spectral and hardware resources~\cite{9737357,11184506}. In multi-user multiple-input multiple-output (MU-MIMO) ISAC systems, transmit beamforming shapes a desired sensing beampattern while guaranteeing the signal-to-interference-plus-noise ratio (SINR) requirement of each communication user (CU)~\cite{8288677}.

\begin{figure}[!t]
\centering
\includegraphics[width=0.8\columnwidth]{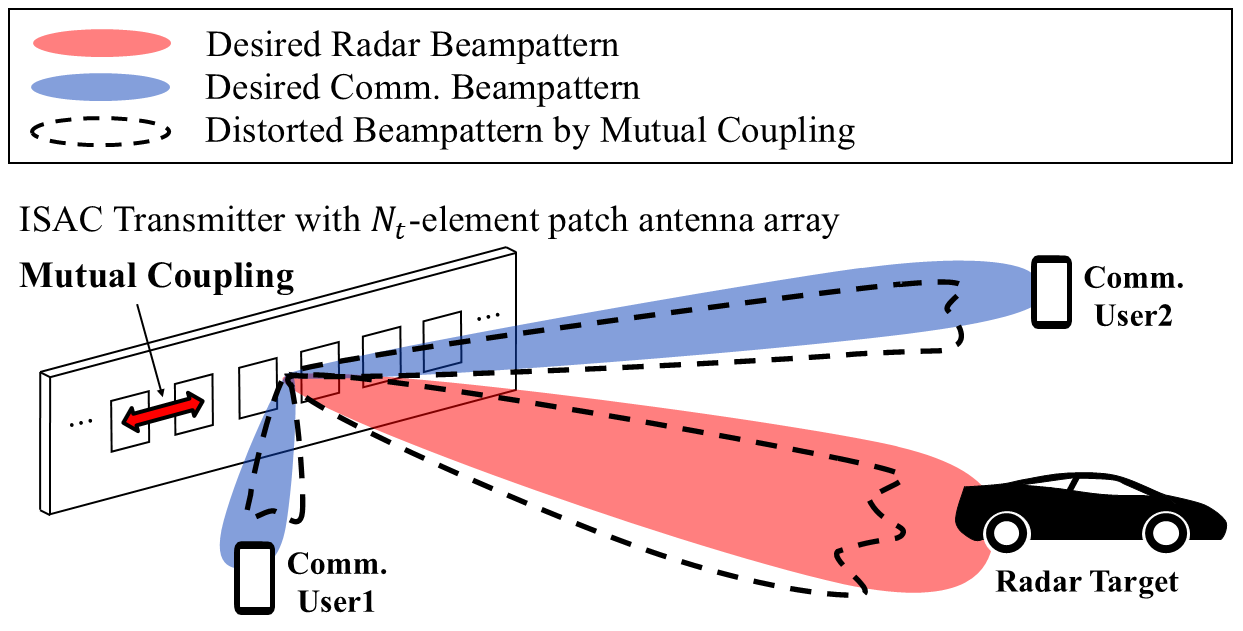}%
\caption{MU-MIMO ISAC beamforming with an $\Nt$-element patch array, where mutual coupling degrades both sensing and communication.}
\label{fig:scenario}
\end{figure}

Conventional ISAC beamforming is built on the ideal steering vector of the antenna array, which assumes no electromagnetic (EM) interaction between elements. In the antenna arrays of practical ISAC transmitters, however, this assumption breaks down. The mutual coupling (MC) effect between closely spaced elements becomes non-negligible, as the fields radiated by each element couple into the adjacent elements~\cite{1143128}.

This MC-induced distortion can be mitigated by MC-compensated beamforming, which replaces the ideal array response with the coupled one obtained from an MC model of the antenna array. Such a model, however, relies on idealized physical conditions that hold only approximately in practice, leaving a residual MC uncertainty. Under this uncertainty, in MIMO radar sensing, MC distorts the transmit beampattern and degrades target detection and parameter estimation, while in MU-MIMO communication it reduces the beamforming gain at each user equipment (UE) and corrupts the channel state information (CSI), as illustrated in Fig.~\ref{fig:scenario}.

Although MC-compensated designs replace the ideal array response with the one predicted by a physics-based MC model~\cite{11175425,11543322}, they evaluate the sensing beampattern and the channel at a fixed response, offering no robustness against the residual MC error. 
Existing robust ISAC designs consider CSI uncertainty~\cite{10153696,11072251} and target-location uncertainty~\cite{10666854,10056405}. Their uncertainty sets are centered at the ideal response, without accounting for the physical source of uncertainty. Since MC affects both the steering vector and the communication channel, and varies with the array configuration in reconfigurable antenna arrays~\cite{11022693,11434508}, the uncertainty set should account for the physical MC model. Such ISAC beamforming, robust to the MC uncertainty in both radar sensing and communication, has not yet been investigated.

To address this limitation, we propose a robust MC-compensated ISAC beamforming design that provides a worst-case performance guarantee for both radar sensing and MU-MIMO communication under residual MC errors. The main contributions of this letter are summarized as follows. 
First, we derive the centers and uncertainty bounds of the sensing and communication responses from a physical MC model. The same underlying residual MC perturbation induces both steering-vector and communication-channel errors, allowing the proposed design to account explicitly for residual MC uncertainty rather than assuming ideal uncoupled responses. Second, we formulate worst-case sensing and communication requirements under the MC uncertainty model. By outer-bounding the induced uncertainties with tractable norm-bounded sets and applying the S-procedure, we obtain a convex semidefinite program (SDP) that minimizes an upper bound on the worst-case beampattern matching mean-squared error (MSE) while guaranteeing the worst-case SINR of every CU, solvable with polynomial-order complexity.
\vspace{-3pt}
\section{System Model}
\subsection{Communication Model}
Consider the MU-MIMO ISAC scenario in Fig.~\ref{fig:scenario}, where a transmitter with an $\Nt$-element patch antenna array serves $K$ single-antenna CUs while sensing a radar target within the region-of-interest (ROI). Following the standard ISAC signal model~\cite{11548574}, the transmitted signal is
\begin{equation}
\bx=\mathbf{W}_c\,\mathbf{c}+\mathbf{W}_s\,\mathbf{s},
\end{equation}
where $\mathbf{W}_c=[\bw_{c,1},\ldots,\bw_{c,K}]\in\mathbb{C}^{\Nt\times K}$ collects the communication beamformers, $\mathbf{c}\in\mathbb{C}^{K}$ is the data-symbol vector of the $K$ CUs, and $\mathbf{W}_s\in\mathbb{C}^{\Nt\times\Nt}$ and $\mathbf{s}\in\mathbb{C}^{\Nt}$ are the beamforming matrix and waveform of a dedicated sensing signal. The symbols and waveforms are zero-mean, unit-power, and mutually uncorrelated,
\begin{equation}
\mathbb{E}\big[\mathbf{c}\mathbf{c}^{H}\big]=\bI_{K},\quad
\mathbb{E}\big[\mathbf{s}\mathbf{s}^{H}\big]=\bI_{\Nt},\quad
\mathbb{E}\big[\mathbf{s}\mathbf{c}^{H}\big]=\mathbf{0}.
\label{eq:assump}
\end{equation}
The channel of user $k$ follows the Rician model $\bh_k=\sqrt{\kappa/(\kappa+1)}\,\bh_{\mathrm{LoS},k}+\sqrt{1/(\kappa+1)}\,\bh_{\mathrm{NLoS},k}$, where $\bh_{\mathrm{LoS},k}$ is the line-of-sight (LoS) component along the steering vector toward the user angle $\theta_k$, $\bh_{\mathrm{NLoS},k}\sim\CN(\mathbf{0},\bI_{\Nt})$ is the Rayleigh-scattered non-line-of-sight (NLoS) component, and $\kappa$ is the Rician factor. The received signal of user $k$ is $y_k=\bh_k^{H}\bx+n_k$ with noise $n_k\sim\CN(0,\sigma^2)$. Since the multi-user and sensing components act as interference, the SINR $\gamma_k$ of user $k$ is given by~\cite{11548574}
\begin{equation}
\gamma_k=\frac{\big|\bh_k^{H}\bw_{c,k}\big|^{2}}
{\sum_{i=1,i\neq k}^{K}\big|\bh_k^{H}\bw_{c,i}\big|^{2}
+\big\|\bh_k^{H}\mathbf{W}_s\big\|_2^{2}+\sigma^2}.
\end{equation}
All beamforming designs investigated in the following sections utilize the SINR as the communication performance metric.
\vspace{-3pt}
\subsection{Sensing Model}
For sensing, the transmit array is operated as a MIMO radar system. For a uniform linear array (ULA) with element spacing $d$, the ideal steering vector is
\begin{equation}
\ba(\theta) = \left[\,1,\; e^{j\frac{2\pi d}{\lambda_0}\sin\theta},\; \dots,\;
e^{j\frac{2\pi d}{\lambda_0}(\Nt-1)\sin\theta}\,\right]^T,
\label{eq:ideal}
\end{equation}
where $\lambda_0$ is the carrier wavelength and $d=\lambda_0/2$ is assumed throughout this letter, with $\|\ba(\theta)\|_2=\sqrt{\Nt}$. With the per-user covariance $\bR_{c,k}=\bw_{c,k}\bw_{c,k}^{H}$ and the sensing covariance $\bR_s=\mathbf{W}_s\mathbf{W}_s^{H}$, the total transmit covariance under~\eqref{eq:assump} is
\begin{equation}
\bR=\mathbb{E}\big[\bx\bx^{H}\big]=\mathbf{W}_c\mathbf{W}_c^{H}
+\mathbf{W}_s\mathbf{W}_s^{H}=\sum_{k=1}^{K}\bR_{c,k}+\bR_s,
\label{eq:covsum}
\end{equation}
and the transmit beampattern is 
\begin{equation}
P(\theta)=\ba^{H}(\theta)\,\bR\,\ba(\theta).
\label{eq:bp}
\end{equation}
We minimize the beampattern matching MSE between $P(\theta)$ and a desired sensing beampattern $P_d(\theta)$ over a grid of $M$ angles $\{\theta_m\}_{m=1}^{M}$, and utilize this MSE to formulate the objectives of all beamforming designs in the following sections.

\subsection{Antenna Mutual Coupling Model}
The ideal steering vector~\eqref{eq:ideal} assumes no EM interaction between elements, which breaks down under MC. We therefore model the MC of a patch-element ULA, widely adopted in practical ISAC transmitters, by the physics-based MC model of~\cite{balanis2016antenna}. The MC is set by two admittances: the self-admittance $Y_{\mathrm{self}}\in\mathbb{C}$ of an isolated patch, and the mutual admittance $Y_{\mathrm{mut}}(|p-q|\,d)$ between two patches whose centers are separated by $|p-q|\,d$, where $p,q\in\{1,\dots,\Nt\}$ index the array elements. These assemble into the admittance matrix $\mathbf{Y}$ of the $\Nt$-port array, relating the port currents to the port voltages, with~\cite{balanis2016antenna}
\begin{equation}
[\mathbf{Y}]_{pq} =
\begin{cases}
Y_{\mathrm{self}}, & p = q, \\
Y_{\mathrm{mut}}(|p-q|\,d), & p \neq q .
\end{cases}
\label{eq:Ymat}
\end{equation}
Inverting it yields the open-circuit impedance matrix $\mathbf{Z}_{\mathrm{oc}} = \mathbf{Y}^{-1}$, where the inversion automatically accounts for all multiple-scattering paths between elements. Each port is then terminated with the load impedance $Z_{\mathrm{L}}$, giving the MC matrix of the load-terminated array as~\cite{1143128}
\begin{equation}
\bZ = (Z_{\mathrm{self}} + Z_{\mathrm{L}})\left(\mathbf{Z}_{\mathrm{oc}} + Z_{\mathrm{L}}\,\mathbf{I}_{\Nt}\right)^{-1},
\label{eq:zraw}
\end{equation}
where $Z_{\mathrm{self}} = 1/Y_{\mathrm{self}}$ is the self-impedance. 
Both $\mathbf{Y}$ and $\bZ$ depend on the resonant patch length $L$~\cite{balanis2016antenna}. In practice, $\mathbf{Y}$ is measured at the antenna ports, and $\bZ$ follows from $\mathbf{Y}$ and $Z_{\mathrm{L}}$ through~\eqref{eq:zraw}.
For given port voltages, $(\mathbf{Z}_{\mathrm{oc}}+Z_{\mathrm{L}}\,\mathbf{I}_{\Nt})^{-1}$ returns the currents of the coupled array and $(Z_{\mathrm{self}}+Z_{\mathrm{L}})^{-1}\mathbf{I}_{\Nt}$ those of the uncoupled array. $\bZ$ therefore maps the uncoupled currents to the coupled ones that set the radiated pattern, and reduces to $\mathbf{I}_{\Nt}$ when coupling is absent ($\mathbf{Z}_{\mathrm{oc}}=Z_{\mathrm{self}}\mathbf{I}_{\Nt}$).

The MC matrix $\bZ$ is normalized as $\bZb=\bZ/\|\bZ\|_2$, only to separate the spatial distortion induced by MC from an overall scalar gain, and is used consistently in all MC-aware designs. Replacing the steering vector~\eqref{eq:ideal} by its coupled form gives
\begin{equation}
\aMC(\theta)=\bZb\,\ba(\theta),
\label{eq:amc}
\end{equation}
on which MC-compensated beamforming evaluates the beampattern~\eqref{eq:bp}. The same matrix acts on the channel~\cite{1310320}, giving
\begin{equation}
\hMC=\bZb\,\bh_k.
\label{eq:hmc}
\end{equation}
Both multiplications are linear because the coupling acts on the element currents, which set the field radiated toward every angle and along every path alike.
\section{Problem Formulation and Proposed Solution}
\subsection{Problem Formulation}
In this section, we present a robust beamforming design under a residual MC error. We start from the original problem $(\mathcal{P}_0)$~\cite{9124713}, which minimizes the sensing beampattern matching MSE subject to the SINR constraint of every CU,
\begin{subequations}
\begin{align}
\min_{\mathbf{W}_c,\mathbf{W}_s}\quad
& f_r \\
\text{s.t.}\quad & \gamma_k\ge\Gamma,\ \forall k, \label{eq:p0sinr}\\
& [\bR]_{nn}=\Pt/\Nt,\ \forall n, \label{eq:p0pwr}
\end{align}
\end{subequations}
where $\Gamma$ is the SINR target and~\eqref{eq:p0pwr} is the per-antenna power constraint. The sensing metric $f_r$ is the beampattern matching MSE given by
\begin{equation}
f_r=\frac{1}{M}\sum_{m=1}^{M}\big(P_d(\theta_m)-\ba^{H}(\theta_m)\bR\,\ba(\theta_m)\big)^2.
\label{eq:fr}
\end{equation}

Problem $(\mathcal{P}_0)$ is non-convex due to the quadratic SINR constraint~\eqref{eq:p0sinr}, and is relaxed to a convex problem by semidefinite relaxation (SDR) over the covariances of~\eqref{eq:covsum}, under which~\eqref{eq:p0sinr} becomes $\big(1+\tfrac{1}{\Gamma}\big)\bh_k^{H}\bR_{c,k}\bh_k\ge\bh_k^{H}\bR\,\bh_k+\sigma^{2}$. The relaxed MC-uncompensated design $(\mathcal{P}_1)$ is then formulated as\begin{subequations}\label{eq:p1all}
\begin{align}
\min_{\bR,\{\bR_{c,k}\}}\quad
& \frac{1}{M}\sum_{m=1}^{M}\Big(P_d(\theta_m)
-\ba^{H}(\theta_m)\bR\,\ba(\theta_m)\Big)^{2}
\label{eq:p1}\\
\text{s.t.}\quad
& \big(1+\tfrac{1}{\Gamma}\big)\bh_k^{H}\bR_{c,k}\bh_k
\ge\bh_k^{H}\bR\bh_k+\sigma^2,\ \forall k, \label{eq:sinrcon}\\
& [\bR]_{nn}=\Pt/\Nt,\ \forall n, \label{eq:pwrcon}\\
& \bR-\textstyle\sum_{k}\bR_{c,k}\succeq\mathbf{0},\
\bR_{c,k}\succeq\mathbf{0},\ \forall k, \label{eq:psdcon}
\end{align}
\end{subequations}
which is convex. Here, $\bR-\sum_{k}\bR_{c,k}\succeq\mathbf{0}$ preserves the covariance structure~\eqref{eq:covsum}, ensuring a valid sensing covariance $\bR_s\succeq\mathbf{0}$. The relaxation is tight: since the objective and~\eqref{eq:pwrcon} depend only on $\bR$, and $\bR_{c,k}$ enters~\eqref{eq:sinrcon} only through $\bh_k^{H}\bR_{c,k}\bh_k$, rank-one beamformers attaining the same objective can be constructed from the solution of~$(\mathcal{P}_1)$~\cite{9124713}.

The MC-compensated design $(\mathcal{P}_2)$ has the same form as~\eqref{eq:p1all}, with the ideal $\ba(\theta)$ and $\bh_k$ replaced throughout by the coupled $\aMC(\theta)$ of~\eqref{eq:amc} and $\hMC$ of~\eqref{eq:hmc}, so that both the beampattern and the channel are evaluated on the coupled array. It is noted that both $(\mathcal{P}_1)$ and $(\mathcal{P}_2)$ are non-robust designs, as they evaluate the beampattern and the channel at a fixed array response and offer no guarantee against the residual uncertainty left by the MC model.

Therefore, we propose a robust MC-compensated design that places the residual error on the MC model, via the norm-bounded error model of~\cite{5765479},
\begin{equation}
\bZrob=\bZb+\bDel,\qquad \|\bDel\|_2\le\varepsilon,
\label{eq:zrob-def}
\end{equation}
where $\bDel$ is the residual coupling matrix and $\varepsilon$ its spectral-norm bound, expressed on a decibel scale as $\eta=10\log_{10}\varepsilon^{2}$. In practice, $\varepsilon$ is estimated from the measured resonant patch-length deviation $\delta L$ as $\varepsilon\approx|\delta L|\,\|\partial\bZb/\partial L\|_2$.

\subsubsection{Radar beampattern uncertainty}
For the radar sensing uncertainty, replacing $\bZb$ in~\eqref{eq:amc} by
$\bZrob$ yields the steering vector under the residual MC error,
\begin{equation}
\arob(\theta_m)=\bZrob\,\ba(\theta_m)
=\aMC(\theta_m)+\bDel\,\ba(\theta_m).
\label{eq:arob-def}
\end{equation}
Writing its residual as $\mathbf{e}_{s,m}=\bDel\,\ba(\theta_m)$, the norm
bound in~\eqref{eq:zrob-def} together with $\|\ba(\theta_m)\|_2=\sqrt{\Nt}$
from~\eqref{eq:ideal} gives the sensing error bound $\epss$,
\begin{equation}
\|\mathbf{e}_{s,m}\|_2\le\|\bDel\|_2\,\|\ba(\theta_m)\|_2
\le\sqrt{\Nt}\,\varepsilon=\epss,\qquad\forall m.
\label{eq:es-bound}
\end{equation}
Replacing $\ba(\theta_m)$ by $\arob(\theta_m)$ in the objective~\eqref{eq:p1} therefore makes the matching MSE a function of the residual.
\subsubsection{Communication channel uncertainty}
For the communication channel uncertainty, replacing $\bZb$ in~\eqref{eq:hmc} by $\bZrob$ gives the channel of user $k$ under the residual error,
\begin{equation}
\hrob=\bZrob\,\bh_k=\hMC+\bDel\,\bh_k.
\label{eq:hrob-def}
\end{equation}
Writing its residual as $\mathbf{e}_{c,k}=\bDel\,\bh_k$, the same norm bound gives the per-user error bound $\epsc$,
\begin{equation}
\|\mathbf{e}_{c,k}\|_2\le\|\bDel\|_2\big\|\bh_k\big\|_2
\le\varepsilon\big\|\bh_k\big\|_2=\epsc,\qquad\forall k.
\label{eq:eck-bound}
\end{equation}
The bound $\epsc$ varies from user to user through $\|\bh_k\|_2$, and is computable at the transmitter from $\varepsilon$ and the available CSI.
Replacing $\bh_k$ by $\hrob$ in the SINR constraint~\eqref{eq:sinrcon} makes it depend on the residual error.

The proposed robust design $(\mathcal{P}_3)$ minimizes the worst-case sensing beampattern matching MSE subject to $\gamma_k\ge\Gamma$ for every communication user and any residual within the error bound. Substituting $\arob(\theta_m)$ of~\eqref{eq:arob-def} for $\ba(\theta_m)$ in~\eqref{eq:fr} gives the matching MSE under the residual error, $f_{r,\mathrm{rob}}=\frac{1}{M}\sum_{m=1}^{M}\big(P_d(\theta_m)-\arob^{H}(\theta_m)\bR\,\arob(\theta_m)\big)^{2}$, and $(\mathcal{P}_3)$ is formulated as
\vspace{-2pt}
\begin{subequations}
\begin{align}
\min_{\bR,\{\bR_{c,k}\}}\;\max_{\|\bDel\|_2\le\varepsilon}\;
& f_{r,\mathrm{rob}}
\label{eq:p3obj}\\
\text{s.t.}\quad
& \big(1+\tfrac{1}{\Gamma}\big)\hrob^{H}\bR_{c,k}\hrob\nonumber\\
&\quad\ge\hrob^{H}\bR\,\hrob+\sigma^2,\ \forall k,
\label{eq:p3sinr}\\
& \eqref{eq:pwrcon},\eqref{eq:psdcon}, \notag
\end{align}
\end{subequations}
where the common residual error $\bDel$ enters the objective through $\mathbf{e}_{s,m}$ and the SINR constraint through $\mathbf{e}_{c,k}$.
\vspace{-3pt}
\subsection{Proposed Solution for Robust Design $(\mathcal{P}_3)$}
In $(\mathcal{P}_3)$, the objective~\eqref{eq:p3obj} carries a maximization over the residual $\bDel$ with $\|\bDel\|_2\le\varepsilon$, and the SINR constraint~\eqref{eq:p3sinr} is semi-infinite over the same set. Therefore, we develop a tractable convex SDP $(\mathcal{P}_3')$ that guarantees the worst-case performance of $(\mathcal{P}_3)$. Since the induced errors $\{\mathbf e_{s,m}\}$ and $\{\mathbf e_{c,k}\}$ share the common residual matrix $\bDel$, they are generally correlated. For a tractable robust formulation, we outer-bound each error by its norm ball in~\eqref{eq:es-bound} and~\eqref{eq:eck-bound} and robustify these sets independently, which enlarges the original uncertainty set and yields a tractable approximation of $(\mathcal{P}_3)$.

To turn the min--max objective~\eqref{eq:p3obj} into a minimization, we introduce the auxiliary variable $S_m$ to upper-bound the worst-case matching error,
\vspace{-2pt}
\begin{equation}
S_m\ \ge\ \max_{\|\mathbf{e}_{s,m}\|_2\le\epss}
\big|\,P_d(\theta_m)-\arob^{H}(\theta_m)\bR\,\arob(\theta_m)\big|.
\label{eq:tbound}
\end{equation}
Then the min--max objective of $(\mathcal{P}_3)$ is upper-bounded by $\tfrac{1}{M}\sum_{m}S_m^{2}$, which is minimized subject to~\eqref{eq:tbound}.
Writing the absolute value in~\eqref{eq:tbound} as a two-sided bound and substituting $\arob(\theta_m)=\aMC(\theta_m)+\mathbf{e}_{s,m}$ from~\eqref{eq:arob-def}
then yields the two constraints
\vspace{-2pt}
\begin{subequations}\label{eq:signsplit}
\begin{multline}
\big(\aMC(\theta_m)+\mathbf{e}_{s,m}\big)^{H}\bR
\big(\aMC(\theta_m)+\mathbf{e}_{s,m}\big)\\
\ge P_d(\theta_m)-S_m,
\label{eq:undershoot}
\end{multline}
\begin{multline}
\big(\aMC(\theta_m)+\mathbf{e}_{s,m}\big)^{H}(-\bR)
\big(\aMC(\theta_m)+\mathbf{e}_{s,m}\big)\\
\ge -P_d(\theta_m)-S_m,
\label{eq:overshoot}
\end{multline}
\end{subequations}
whose sum yields $S_m\ge0$, so the bound is automatically nonnegative.

For the SINR constraint~\eqref{eq:p3sinr}, since the residual enters only through the channel, we collect its terms into a quadratic form in the channel and substitute $\hrob=\hMC+\mathbf{e}_{c,k}$ from~\eqref{eq:hrob-def}, giving
\vspace{-2pt}
\begin{equation}
\big(\hMC+\mathbf{e}_{c,k}\big)^{H}
\Big[\big(1+\tfrac{1}{\Gamma}\big)\bR_{c,k}-\bR\Big]
\big(\hMC+\mathbf{e}_{c,k}\big)\ \ge\ \sigma^{2},
\label{eq:cmargin}
\end{equation}
required for every $\|\mathbf{e}_{c,k}\|_2\le\epsc$.

The beampattern constraints~\eqref{eq:signsplit} and the SINR constraint~\eqref{eq:cmargin} remain semi-infinite, but now share one common algebraic form,
\begin{equation}
(\mathbf{x}+\mathbf{e})^{H}\bQ(\mathbf{x}+\mathbf{e})\ \ge\ \beta,
\qquad\forall\,\|\mathbf{e}\|_2\le r.
\label{eq:template}
\end{equation}
Here, $\mathbf{x}$ is the nominal vector, $\bQ$ is a Hermitian matrix, $\beta$ is the required level, and $r$ is the error bound of the uncertainty set. The beampattern constraints~\eqref{eq:undershoot} and~\eqref{eq:overshoot} take this form with $\mathbf{x}=\aMC(\theta_m)$, $\bQ=\pm\bR$, $\beta=\pm P_d(\theta_m)-S_m$ and $r=\epss$, and the SINR constraint~\eqref{eq:cmargin} with $\mathbf{x}=\hMC$,
$\bQ=\big(1+\tfrac{1}{\Gamma}\big)\bR_{c,k}-\bR$, $\beta=\sigma^{2}$ and
$r=\epsc$, where $\epsc>0$ for any nonzero $\bh_k$. In both cases, the constraint and the norm bound are quadratic in $\mathbf{e}_{s,m}$ and $\mathbf{e}_{c,k}$, respectively, and the following lemma converts this pair of quadratic conditions into a linear matrix inequality (LMI).

\begin{lemma}\label{lem:sproc}
Let $\bQ=\bQ^{H}$ and $r>0$. Then~\eqref{eq:template} holds if and only if there exists $\lambda\ge0$ such that
\begin{equation}
\mathbf{L}(\bQ,\mathbf{x},\beta,r,\lambda)=
\begin{bmatrix}
\bQ+\lambda\bI_{\Nt} & \bQ\mathbf{x}\\
\mathbf{x}^{H}\bQ &
\mathbf{x}^{H}\bQ\mathbf{x}-\beta-\lambda r^{2}
\end{bmatrix}\succeq\mathbf{0},
\label{eq:Ltemplate}
\end{equation}
an LMI in $\lambda$ and the entries of $(\bQ,\beta)$ for fixed $\mathbf{x}$.
\end{lemma}

\begin{IEEEproof}
Expanding the left side of~\eqref{eq:template} gives the quadratic function $f_0(\mathbf{e})=\mathbf{e}^{H}\bQ\mathbf{e}+2\Rea\big((\bQ\mathbf{x})^{H}\mathbf{e}\big)+\mathbf{x}^{H}\bQ\mathbf{x}-\beta$, and the norm bound is $f_1(\mathbf{e})=r^{2}-\mathbf{e}^{H}\mathbf{e}\ge0$. Since $f_1(\mathbf{0})=r^{2}>0$, the S-procedure, lossless for complex-valued quadratics, states that $f_0\ge0$ on $\{f_1\ge0\}$ if and only if there exists $\lambda\ge0$ with $f_0(\mathbf{e})-\lambda f_1(\mathbf{e})\ge0$ for all $\mathbf{e}$. Writing this difference as $[\mathbf{e}^{H},1]\,\mathbf{L}(\bQ,\mathbf{x},\beta,r,\lambda)\,[\mathbf{e}^{T},1]^{T}$ shows, by the standard homogenization argument, that its nonnegativity for all $\mathbf{e}$ is equivalent to~\eqref{eq:Ltemplate}.
\end{IEEEproof}

\begin{figure*}[!t]
\centering
\subfloat[]{\includegraphics[width=0.222\linewidth]{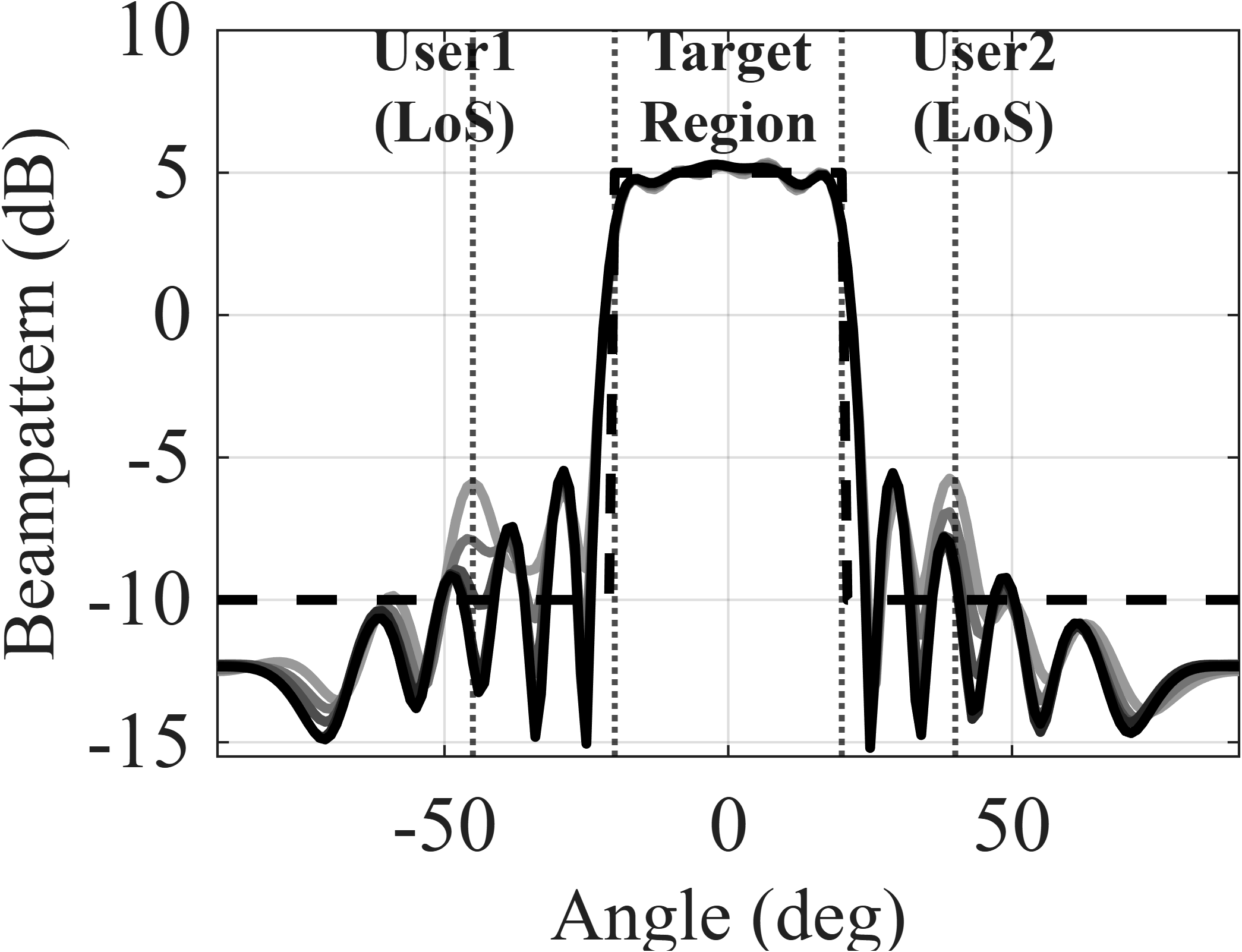}}
\hfil
\subfloat[]{\includegraphics[width=0.205\linewidth]{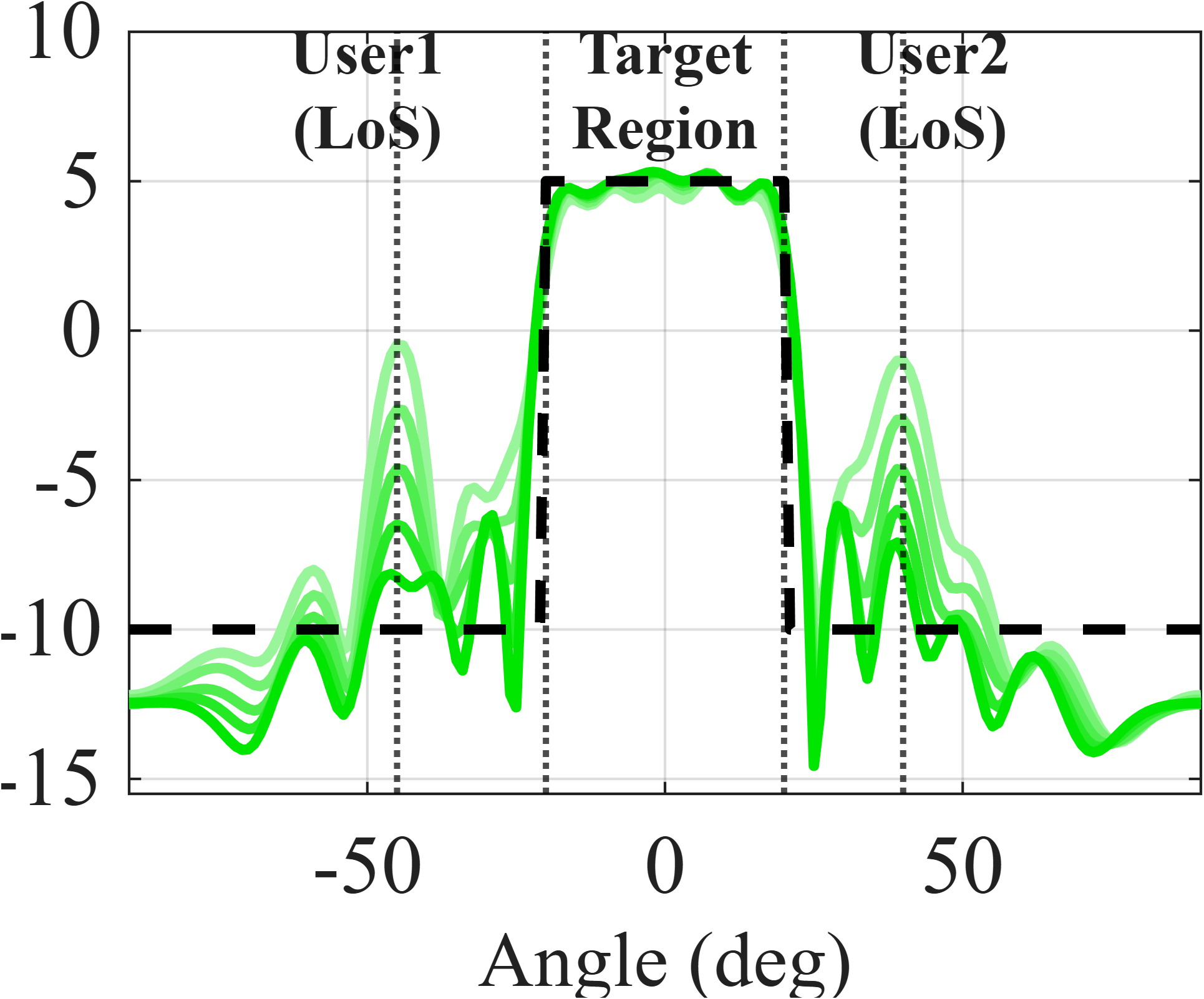}}
\hfil
\subfloat[]{\includegraphics[width=0.205\linewidth]{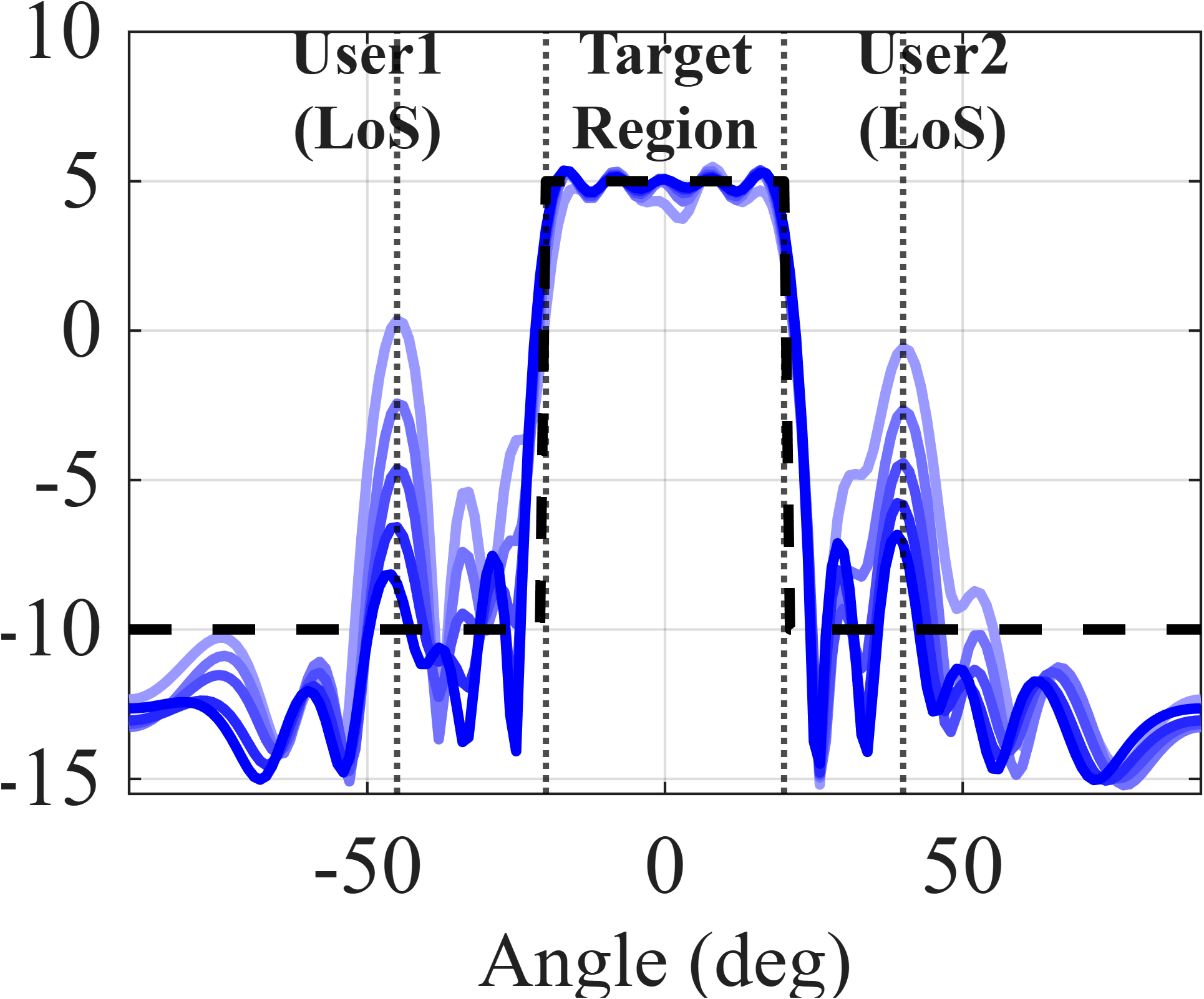}}
\hfil
\subfloat[]{\includegraphics[width=0.205\linewidth]{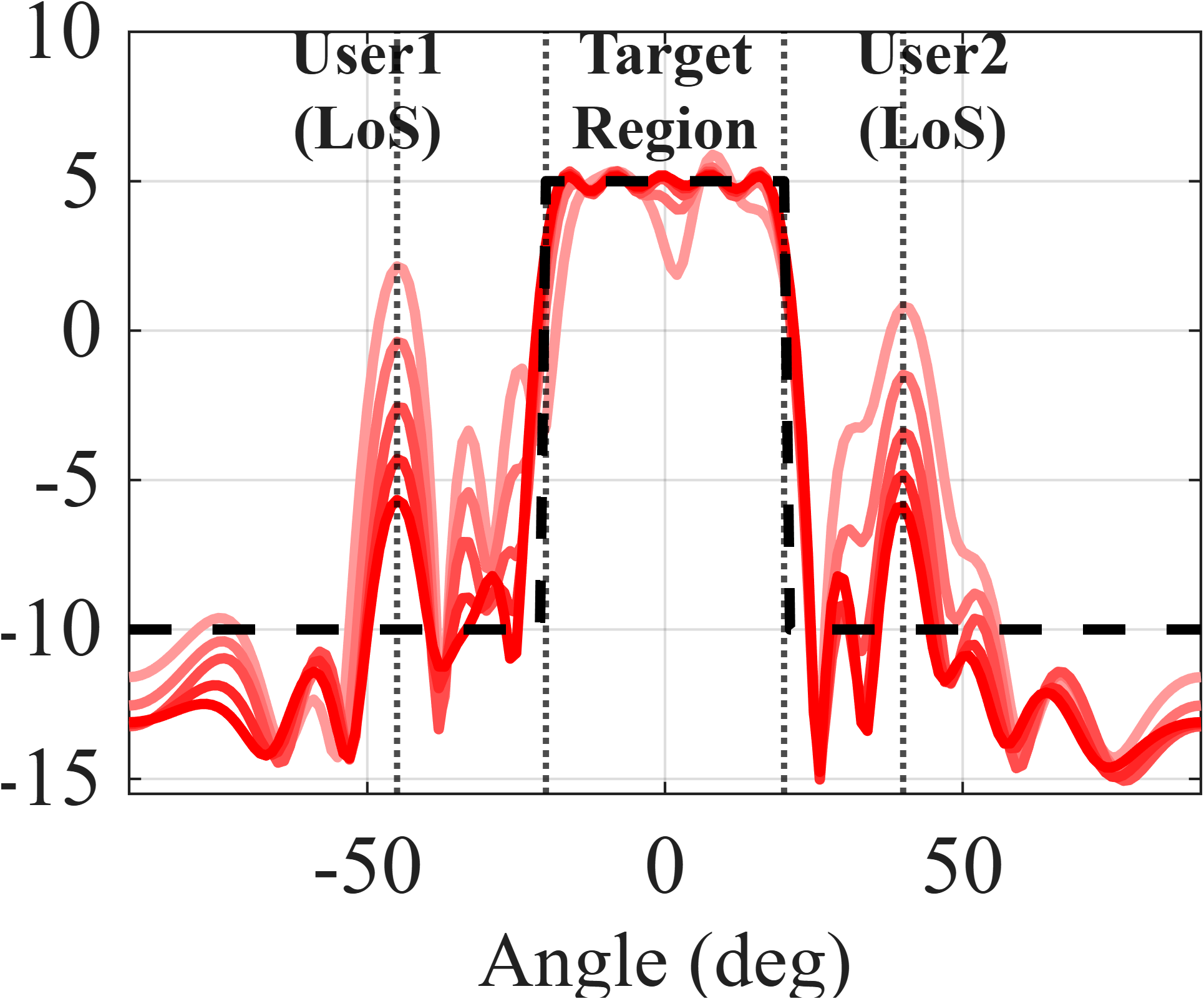}}
\hfil
\raisebox{\dimexpr\baselineskip+2pt\relax}{\includegraphics[width=0.1\linewidth]{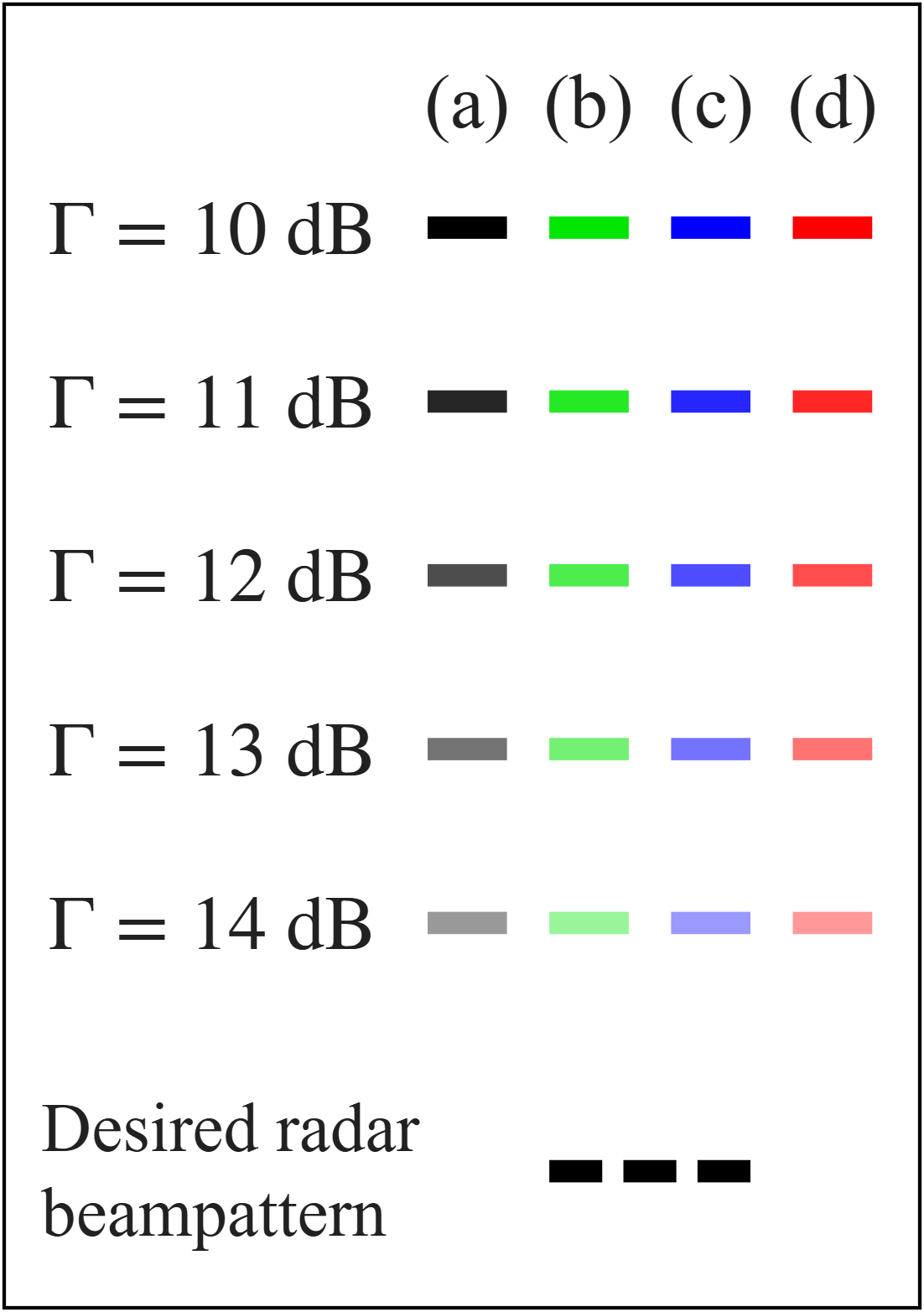}}
\caption{Beampatterns of the (a)~MC-uncompensated, (b)~CSI-robust~\cite{10153696}, (c)~MC-compensated, and (d)~robust MC-compensated designs.}
\label{fig:bp}
\end{figure*}

Applying Lemma~\ref{lem:sproc} to the beampattern constraints~\eqref{eq:undershoot} and~\eqref{eq:overshoot} with multipliers $\nu_m\ge0$ and $\mu_m\ge0$ gives
\begin{align}
\bPhi_m&=\mathbf{L}\big(\bR,\aMC(\theta_m),P_d(\theta_m)-S_m,\epss,\nu_m\big)
\succeq\mathbf{0},
\label{eq:lmi-under}\\
\bPsi_m&=\mathbf{L}\big(-\bR,\aMC(\theta_m),-P_d(\theta_m)-S_m,\epss,\mu_m\big)
\succeq\mathbf{0},
\end{align}
and to the SINR constraint~\eqref{eq:cmargin} with multiplier $\lambda_k\ge0$ gives
\begin{equation}
\bXi_k=\mathbf{L}\Big(\big(1+\tfrac{1}{\Gamma}\big)\bR_{c,k}-\bR,
\hMC,\sigma^{2},\epsc,\lambda_k\Big)\succeq\mathbf{0}.
\label{eq:lmi-comm}
\end{equation}

Since $\bPhi_m$ and $\bPsi_m$ are linear in $(\bR,S_m,\nu_m,\mu_m)$ and $\bXi_k$ is linear in $(\bR,\bR_{c,k},\lambda_k)$, replacing the beampattern constraints~\eqref{eq:signsplit} and the SINR constraint~\eqref{eq:cmargin} by~\eqref{eq:lmi-under}--\eqref{eq:lmi-comm} yields the convex SDP $(\mathcal{P}_3')$. With the optimization variables collected as $\mathcal{V}=\big\{\bR,\{\bR_{c,k}\},\{S_m\},\{\nu_m\},\{\mu_m\},\{\lambda_k\}\big\}$, the proposed robust design $(\mathcal{P}_3')$ is formulated as
\begin{equation}
\begin{aligned}
\min_{\mathcal{V}}\quad
& \frac{1}{M}\sum_{m=1}^{M}S_m^{2}\\
\text{s.t.}\quad
& \bPhi_m\succeq\mathbf{0},\ \bPsi_m\succeq\mathbf{0},\ \nu_m\ge0,\ \mu_m\ge0,\ \forall m,\\
& \bXi_k\succeq\mathbf{0},\ \lambda_k\ge0,\ \forall k,\ \eqref{eq:pwrcon},\eqref{eq:psdcon}.
\end{aligned}
\end{equation}

\begin{remark}
The CSI-robust design of~\cite{10153696} centers channel uncertainty at $\bh_k$ with bound $\mu\|\bh_k\|_2$, where $\mu$ is the channel-error percentage, and does not consider steering-vector uncertainty. It corresponds to $(\mathcal{P}_3')$ with $\bZb=\bI_{\Nt}$, $\epss=0$, and $\epsc=\mu\|\bh_k\|_2$ under the same objective and constraints, differing only in the uncertainty set. Since its center $\bh_k$ does not account for the physical MC distortion, it differs from $\hMC$ by $\|(\bZb-\bI_{\Nt})\bh_k\|_2$. Covering~\eqref{eq:hrob-def} therefore requires the bound $\|(\bZb-\bI_{\Nt})\bh_k\|_2+\epsc$, which is larger than $\epsc$ of the proposed design.
\end{remark}

The rank-relaxation method solves $(\mathcal{P}_3')$ by SDR, replacing the rank-one condition with $\bR_{c,k}\succeq\mathbf{0}$. In practice, each communication symbol is transmitted through one beamformer $\bw_{c,k}$, whose covariance $\bR_{c,k}=\bw_{c,k}\bw_{c,k}^{H}$ is rank one. The rank-one projection method thus sets $\bR_{c,k}=p_k\mathbf{u}_k\mathbf{u}_k^{H}$ with $\mathbf{u}_k$ the principal eigenvector of the rank-relaxed solution, and reoptimizes $p_k$ and $\bR$ under the same constraints, yielding $\bw_{c,k}=\sqrt{p_k}\mathbf{u}_k$.
\vspace{-3pt}
\subsection{Computational Complexity}
Problems $(\mathcal{P}_1)$, $(\mathcal{P}_2)$, and $(\mathcal{P}_3')$ are SDPs solved by an interior-point method, whose cost is governed by the number of variables and the sizes of the semidefinite blocks. The former mainly sets the cost of each iteration, the latter the number of iterations. The non-robust designs $(\mathcal{P}_1)$ and $(\mathcal{P}_2)$ optimize the $K{+}1$ covariance matrices $\bR,\{\bR_{c,k}\}\in\mathbb{H}^{\Nt}$ and $M$ residuals under $K{+}1$ blocks of size $\Nt$, giving the per-solve complexity $\mathcal{O}\big(\sqrt{(K{+}1)\Nt}\,\big((K{+}1)\Nt^{2}+M\big)^{3} \log(1/\epsilon)\big)$ for accuracy $\epsilon$. The robust design $(\mathcal{P}_3')$ adds only $3M{+}K$ scalar variables, so each iteration costs the same order, but its worst-case bounds impose one pair $\bPhi_m,\bPsi_m$ per angle and one $\bXi_k$ per user, adding $2M{+}K$ blocks of size $\Nt{+}1$ that raise the iteration count to $\mathcal{O}\big(\sqrt{M\Nt}\big)$. The complexity therefore remains polynomial of the same order in $\Nt$, higher only by the factor $\sqrt{M/(K{+}1)}$, the price of enforcing the worst-case bounds at every angle and user.
\vspace{-3pt}
\section{Numerical Results}
We consider an $\Nt=16$-element patch array with $d=\lambda_0/2$. Each rectangular patch ($W=0.49\lambda_0$, resonant length $L=0.44\lambda_0$) is terminated with $Z_{\mathrm{L}}=50~\Omega$, and $Y_{\mathrm{self}}$, $Y_{\mathrm{mut}}$ follow the physics-based MC model~\cite{balanis2016antenna}, from which $\bZ$ is pre-calculated through~\eqref{eq:Ymat} and~\eqref{eq:zraw}. The array serves $K=2$ users at $-45^\circ$ and $40^\circ$ over Rician channels ($\kappa=1$). The desired pattern $P_d(\theta)$ is $5$~dB over $[-20^\circ,20^\circ]$ with a $-10$~dB floor elsewhere, on a uniform grid of $M=91$ angles. All beamforming designs are evaluated under the same per-antenna transmit-power constraint. The beampatterns are normalized to the radiated power of the desired beampattern for direct comparison of relative sidelobe levels. We set $\Pt/\sigma^2=11$~dB and $\eta=-30$~dB. Unless stated otherwise, the realized mismatch level is $\tilde{\eta}=-12$~dB. Each realization of $\bDel$ has independent complex-Gaussian entries scaled to $\|\bDel\|_2=\tilde{\varepsilon}$, where $\tilde{\eta}=10\log_{10}\tilde{\varepsilon}^{2}$; the same realization perturbs both the steering vectors in~\eqref{eq:arob-def} and the channels. The proposed design in Figs.~\ref{fig:bp}--\ref{fig:multiuser} is evaluated on the solution $\{\bR_{c,k},\bR\}$ of the rank-relaxation method. When the rank-one projection method is applied instead, the beampattern matching MSE of the proposed design at $\tilde{\eta}=-12$~dB increases by at most $0.03$~dB over the whole range of $\Gamma$. For the CSI-robust design~\cite{10153696}, the channel uncertainty bound is set to $\mu\|\bh_k\|_2$ with $\mu=0.1$ as in~\cite{10153696}. Over $1{,}000$ Monte Carlo realizations we report the mean beampattern, the $95$th-percentile matching MSE, and the $5$th-percentile achieved SINR, where the tail percentiles reflect the worst-case-oriented nature of the proposed design.
\begin{figure}[!t]
\centering
\subfloat[]{\includegraphics[width=0.338\columnwidth]{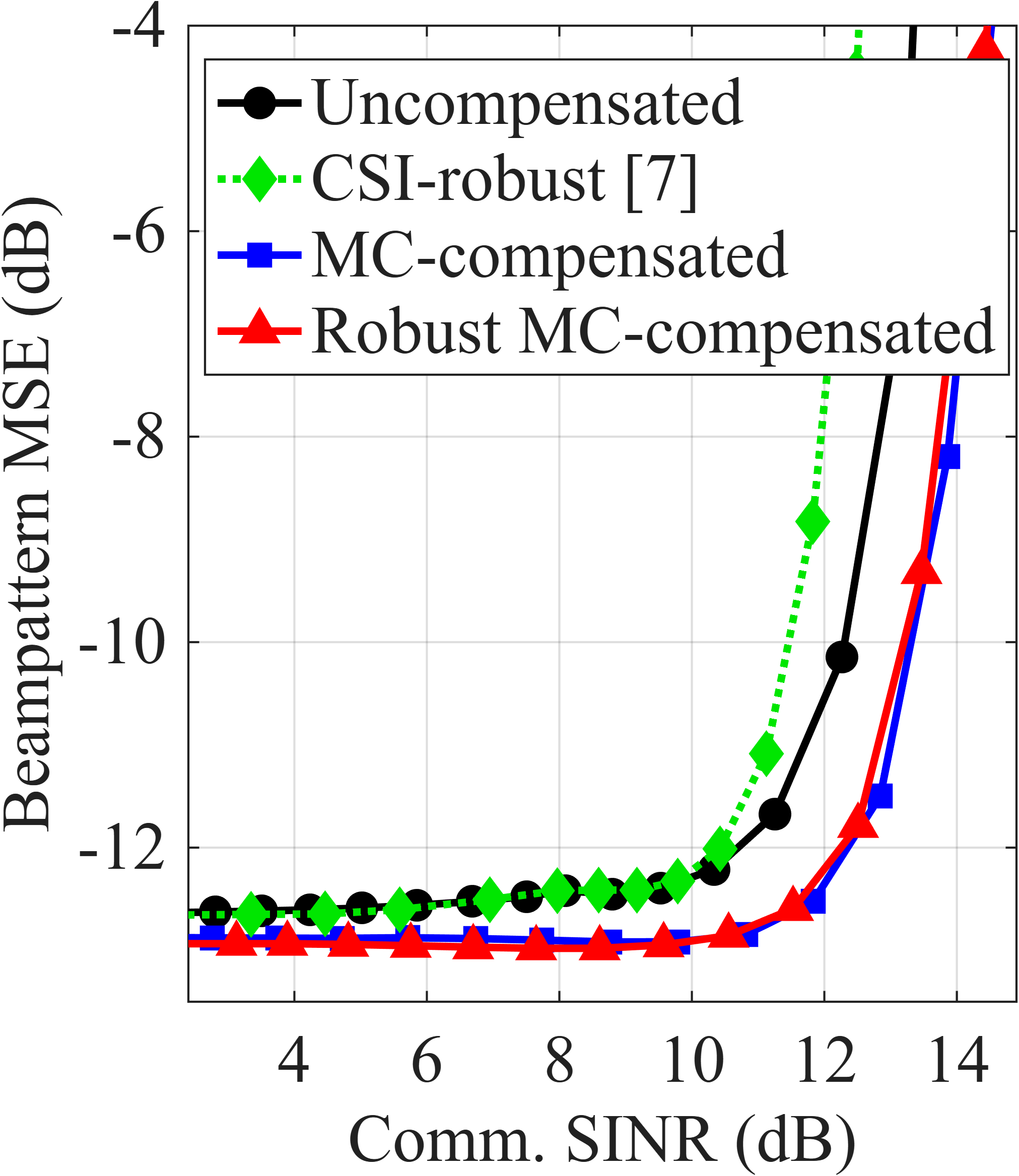}}
\hfil
\subfloat[]{\includegraphics[width=0.312\columnwidth]{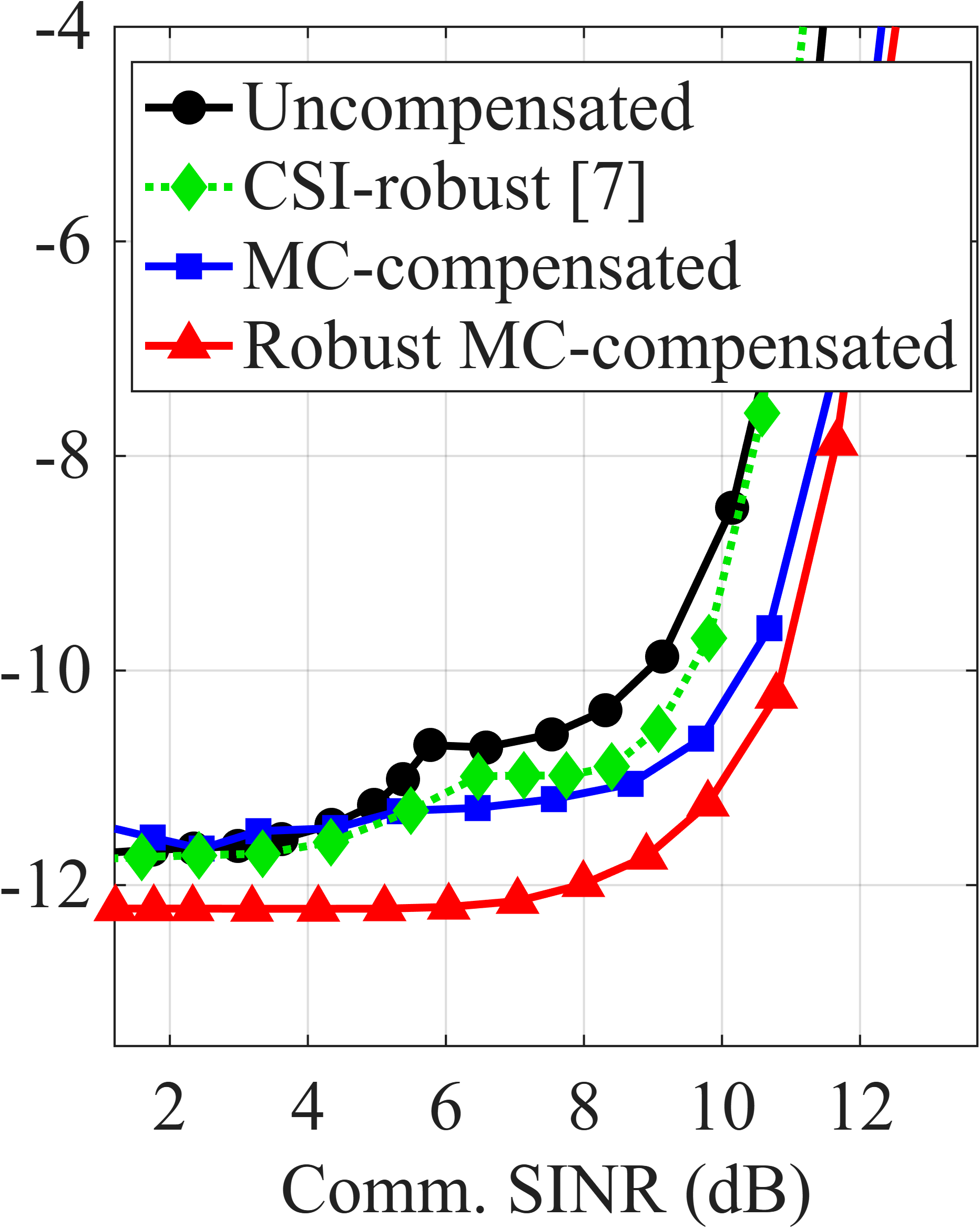}}
\hfil
\subfloat[]{\includegraphics[width=0.312\columnwidth]{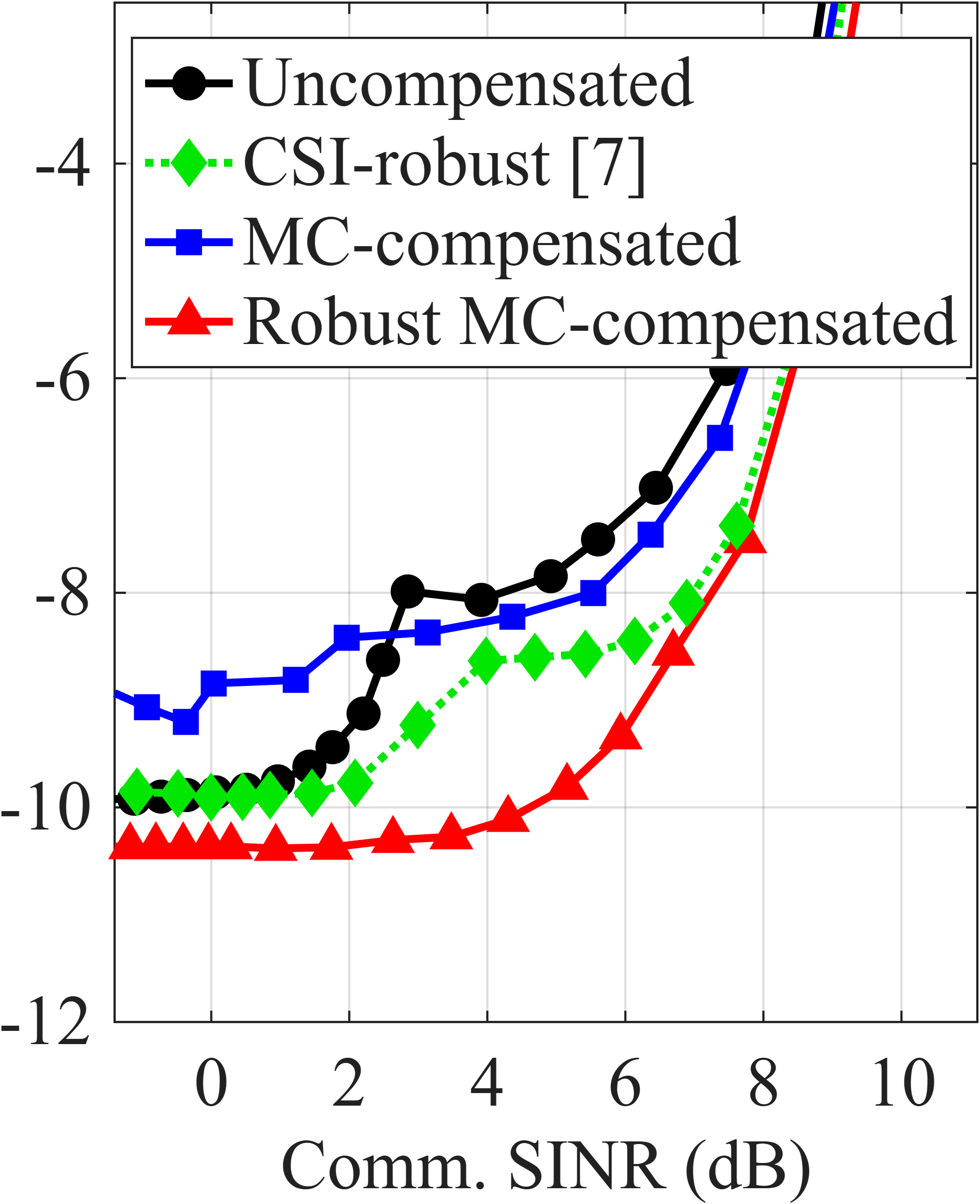}}
\caption{Beampattern matching MSE versus achieved SINR trade-off at $\tilde{\eta}=$ (a)~$-30$, (b)~$-12$, and (c)~$-6$~dB.}
\label{fig:tradeoff}
\end{figure}

\begin{figure}[!t]
\centering
\includegraphics[width=0.86\columnwidth]{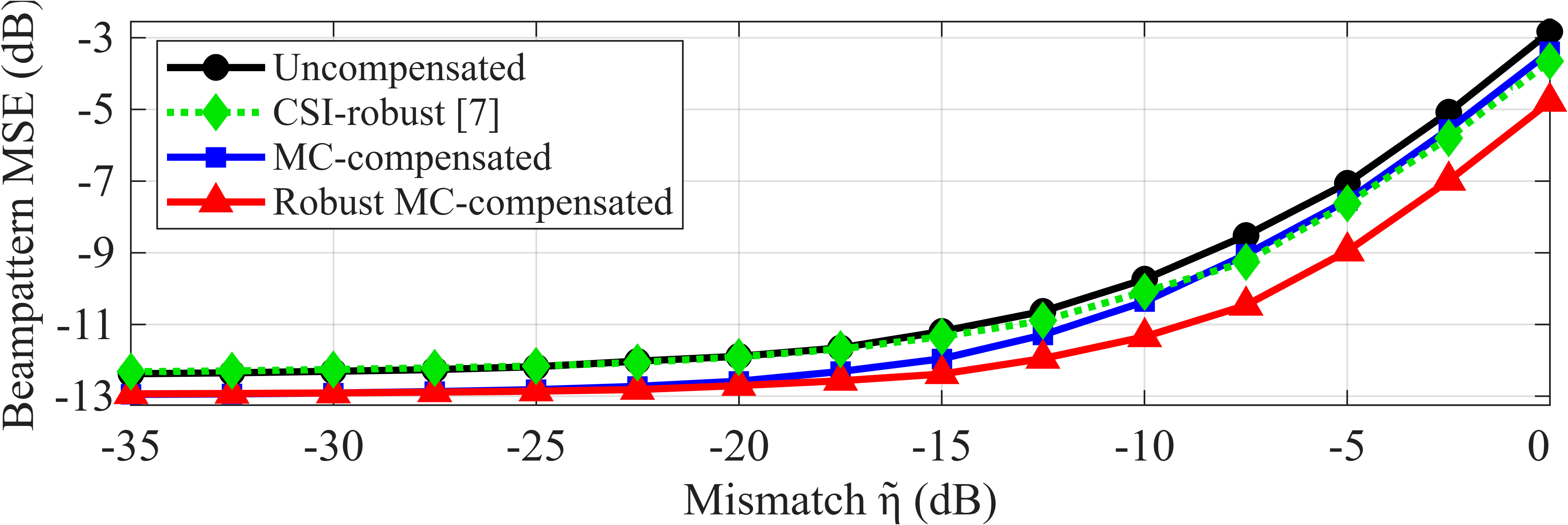}
\caption{Beampattern matching MSE versus the realized mismatch $\tilde{\eta}$.}
\label{fig:mse}
\end{figure}

\begin{figure}[!t]
\centering
\includegraphics[width=0.86\columnwidth]{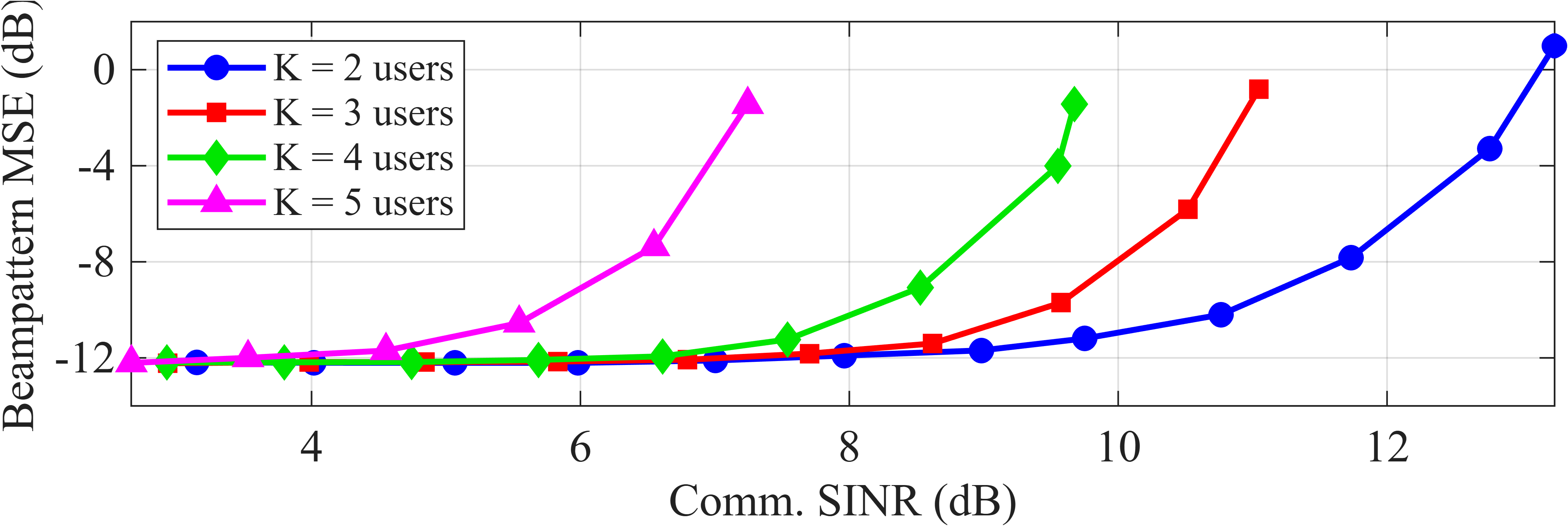}
\caption{Trade-off of the robust MC-compensated design for $K=2$--$5$ users.}
\label{fig:multiuser}
\end{figure}

Fig.~\ref{fig:bp} shows the beampatterns as $\Gamma$ is swept from $10$ to $14$~dB. The proposed design attains the highest achieved SINR at the same SINR target while keeping a low sidelobe level. Fig.~\ref{fig:tradeoff} shows the trade-off between beampattern matching MSE and achieved SINR at $\tilde{\eta}=-30$, $-12$, and $-6$~dB. Fig.~\ref{fig:tradeoff}(a) corresponds to the design bound $\tilde{\eta}=\eta$, and Fig.~\ref{fig:tradeoff}(b) and (c) show that beyond it the proposed design attains the lowest MSE at the same achieved SINR. Fig.~\ref{fig:mse} compares the MSE of the designs as $\tilde{\eta}$ varies within and beyond $\eta$, where the proposed design attains a lower MSE with a gap that widens as $\tilde{\eta}$ increases. Fig.~\ref{fig:multiuser} shows the trade-off of the robust design for $K=2$--$5$ users, adding users at $-60^\circ$, $55^\circ$, and $75^\circ$ to the two at $-45^\circ$ and $40^\circ$. Up to moderate SINR, the MSE is unaffected by $K$; beyond that it rises earlier for larger $K$.
\vspace{-3pt}
\section{Conclusion}
In this letter, we have proposed a robust MC-compensated beamforming design that guarantees both the sensing and communication performance of MU-MIMO ISAC transmitters against the residual MC error. Its uncertainty center and bound are derived from the physical MC model, where the same residual MC perturbation induces both steering-vector and communication-channel uncertainties. With tractable norm-bounded approximations of these uncertainties, each resulting worst-case constraint becomes an LMI, giving a convex SDP of the same polynomial-order complexity as non-robust designs. Numerical results show that the proposed robust design attains a lower beampattern matching MSE and a higher achieved SINR than the existing robust design.
\vspace{-3pt}


\begin{thebibliography}{10}
\providecommand{\url}[1]{#1}
\csname url@samestyle\endcsname
\providecommand{\newblock}{\relax}
\providecommand{\bibinfo}[2]{#2}
\providecommand{\BIBentrySTDinterwordspacing}{\spaceskip=0pt\relax}
\providecommand{\BIBentryALTinterwordstretchfactor}{4}
\providecommand{\BIBentryALTinterwordspacing}{\spaceskip=\fontdimen2\font plus
\BIBentryALTinterwordstretchfactor\fontdimen3\font minus \fontdimen4\font\relax}
\providecommand{\BIBforeignlanguage}[2]{{%
\expandafter\ifx\csname l@#1\endcsname\relax
\typeout{** WARNING: IEEEtran.bst: No hyphenation pattern has been}%
\typeout{** loaded for the language `#1'. Using the pattern for}%
\typeout{** the default language instead.}%
\else
\language=\csname l@#1\endcsname
\fi
#2}}
\providecommand{\BIBdecl}{\relax}
\BIBdecl

\bibitem{9737357}
F.~Liu \emph{et~al.}, ``{Integrated Sensing and Communications: Toward Dual-Functional Wireless Networks for {6G} and Beyond},'' \emph{IEEE Journal on Selected Areas in Communications}, vol.~40, no.~6, pp. 1728--1767, 2022.

\bibitem{11184506}
K.~Han, K.~Meng, X.-Y. Wang, and C.~Masouros, ``{Network-Level ISAC Design: State-of-the-Art, Challenges, and Opportunities},'' \emph{IEEE Journal of Selected Topics in Electromagnetics, Antennas and Propagation}, vol.~1, no.~1, pp. 65--83, 2025.

\bibitem{8288677}
F.~Liu, C.~Masouros, A.~Li, H.~Sun, and L.~Hanzo, ``{{MU-MIMO} Communications With {MIMO} Radar: From Co-Existence to Joint Transmission},'' \emph{IEEE Transactions on Wireless Communications}, vol.~17, no.~4, pp. 2755--2770, 2018.

\bibitem{1143128}
I.~Gupta and A.~Ksienski, ``{Effect of mutual coupling on the performance of adaptive arrays},'' \emph{IEEE Transactions on Antennas and Propagation}, vol.~31, no.~5, pp. 785--791, 1983.

\bibitem{11175425}
S.~Zeng \emph{et~al.}, ``{Holographic Beamforming for Integrated Sensing and Communication With Mutual Coupling Effects},'' \emph{IEEE Journal on Selected Areas in Communications}, vol.~44, pp. 480--497, 2026.

\bibitem{11543322}
Y.~Sun \emph{et~al.}, ``{Mutual Coupling-Aware {RIS}-Aided Integrated Sensing and Communication},'' \emph{IEEE Transactions on Communications}, vol.~74, pp. 9452--9467, 2026.

\bibitem{10153696}
Z.~Ren, L.~Qiu, J.~Xu, and D.~W.~K. Ng, ``{Robust Transmit Beamforming for Secure Integrated Sensing and Communication},'' \emph{IEEE Transactions on Communications}, vol.~71, no.~9, pp. 5549--5564, 2023.

\bibitem{11072251}
Z.~Chen, F.~Wang, G.~Han, X.~Wang, and V.~K.~N. Lau, ``{Robust Beamforming Design for Secure Near-Field ISAC Systems},'' \emph{IEEE Wireless Communications Letters}, vol.~14, no.~10, pp. 3089--3093, 2025.

\bibitem{10666854}
W.~Lyu \emph{et~al.}, ``{Dual-Robust Integrated Sensing and Communication: Beamforming Under CSI Imperfection and Location Uncertainty},'' \emph{IEEE Wireless Communications Letters}, vol.~13, no.~11, pp. 3124--3128, 2024.

\bibitem{10056405}
M.~Luan, B.~Wang, Z.~Chang, T.~Hämäläinen, and F.~Hu, ``{Robust Beamforming Design for RIS-Aided Integrated Sensing and Communication System},'' \emph{IEEE Transactions on Intelligent Transportation Systems}, vol.~24, no.~6, pp. 6227--6243, 2023.

\bibitem{11022693}
S.~Yang \emph{et~al.}, ``Toward Intelligent Antenna Positioning: Leveraging DRL for FAS-Aided ISAC Systems,'' \emph{IEEE Internet of Things Journal}, vol.~12, no.~16, pp. 34\,615--34\,618, 2025.

\bibitem{11434508}
T.~Wu \emph{et~al.}, ``Scalable Fluid Antenna Systems: A New Paradigm for Array Signal Processing,'' \emph{IEEE Journal of Selected Topics in Signal Processing}, vol.~20, no.~3, pp. 389--406, 2026.

\bibitem{11548574}
K.~Han, C.~Masouros, T.~Riihonen, and M.~G. Amin, ``{Next-Generation {MIMO} Transceivers for Integrated Sensing and Communications: Unique Security Vulnerabilities and Solutions},'' \emph{Proceedings of the IEEE}, vol. 114, no.~1, pp. 18--51, 2026.

\bibitem{balanis2016antenna}
C.~A. Balanis, \emph{{Antenna Theory: Analysis and Design}}.\hskip 1em plus 0.5em minus 0.4em\relax John Wiley \& Sons, 2016.

\bibitem{1310320}
J.~Wallace and M.~Jensen, ``{Mutual Coupling in {MIMO} Wireless Systems: A Rigorous Network Theory Analysis},'' \emph{IEEE Transactions on Wireless Communications}, vol.~3, no.~4, pp. 1317--1325, 2004.

\bibitem{9124713}
X.~Liu, T.~Huang, N.~Shlezinger, Y.~Liu, J.~Zhou, and Y.~C. Eldar, ``Joint Transmit Beamforming for Multiuser MIMO Communications and MIMO Radar,'' \emph{IEEE Transactions on Signal Processing}, vol.~68, pp. 3929--3944, 2020.

\bibitem{5765479}
T.~Zhang and W.~Ser, ``{Robust Beampattern Synthesis for Antenna Arrays With Mutual Coupling Effect},'' \emph{IEEE Transactions on Antennas and Propagation}, vol.~59, no.~8, pp. 2889--2895, 2011.

\end{thebibliography}
\end{document}